\documentclass[runningheads, envcountsame, a4paper]{llncs}

\usepackage{amssymb}
\usepackage{amsmath}
\usepackage{textcomp}
\usepackage{array}
\usepackage{lineno}
\usepackage{color}
\usepackage{comment}

\usepackage{graphicx}

\newcommand{\ket}[1]{|#1\rangle}
\newcommand{\braket}[2]{\langle #1|#2\rangle}
\newcommand{\Endproof}{\hfill$\Box$\\}

\newtheorem{fact}{Fact}

\newcommand{\uobdd}{UOBDD}
\newcommand{\nuobdd}{NUOBDD}

\title{Nondeterministic unitary OBDDs}

\author{Aida Gainutdinova\inst{1}\fnmsep\thanks{Some parts of this work was done during Gainutdinova's visit to National Laboratory for Scientific Computing (Brazil) in June 2015 supported by CAPES with grant 88881.030338/2013-01.} \and Abuzer Yakary{\i}lmaz\inst{2}\fnmsep\thanks{Partially supported by CAPES with grant 88881.030338/2013-01 and ERC Advanced Grant MQC.}
}

\institute{Kazan Federal University, Kazan, Russia
\\
 \email{aida.ksu@gmail.com}
\and
University of Latvia, Faculty of Computing, R\={\i}ga, Latvia
\\
\email{abuzer@lu.lv}
}

\authorrunning{A. Gainutdinova, A. Yakary\i lmaz} 

\begin{document}
\maketitle

\begin{abstract}
	We investigate the width complexity of nondeterministic unitary OBDDs (NUOBDDs). Firstly, we present a generic lower bound on their widths based on the size of strong 1-fooling sets. Then, we present classically cheap functions that are expensive for NUOBDDs and vice versa by improving the previous gap. We also present a function for which neither classical nor unitary nondeterminism does help. Moreover, based on our results, we present a width hierarchy for NUOBDDs. Lastly, we provide the bounds on the widths of NUOBDDs for the basic Boolean operations negation, union, and intersection. 
\end{abstract}
  
\section {Introduction}

Branching Programs (BPs) are one of the well known computational models, which are important not only theoritically but also practically, such as   hardware verification, model checking and others \cite{Weg00}. The main complexity measures for BP are the size of BP -- its number of nodes and length (time complexity). 
It is well--known that  BPs of polynomial size are equivalent to non-uniform log-space Turing machines. 
 
The important restricted variant of BPs is  Ordered Binary Decision Diagrams (OBDDs), which  are oblivious read-once branching programs \cite{Weg00}. 
Time complexity for OBDD is at most $n$ (the length of the input), and so the natural complexity measure for OBDD is its width. Different   variants of OBDDs such as deterministic, probabilistic, nondeterministic, and quantum  have been considered (e.g. \cite{ak96,NHK00,AGKMP05,SS05,AGKY16A}) and they have been compared in term of their widths. For example, it was shown that randomized OBDDs can be  exponentially more efficient than deterministic and nondeterministic OBDDs \cite{ak96}, and,  quantum OBDDs can be exponentially more effcient  than deterministic and stable probabilistic OBDD and that this bound is tight \cite{AGKMP05}.  In \cite{SS05} some  simple functions were presented such that unitary OBDDs (the known most restricted quantum OBDD) need exponential
size for computing these functions with bounded error, while deterministic OBDDs can represent these functions in linear size. 
Quantum and classical nondeterminism for OBDD models was considered in \cite{AGKY16A}, where the superiority of quantum OBDDs over classical counterparts was shown.
In particular, an explicit function was presented, which  is computed by a quantum nondeterministic OBDD  of constant width, but any classical nondeterministic OBDD  for this function needs non-constant width. 

The OBDDs of constant width can also be considered as a nonuniform analog of one-way finite automata \cite{ag05}. It is well known that classical  nondeterministic automata recognize precisely  regular languages. There are different  variants of nondeterministic quantum finite automata (NQFA) in  literature \cite{nihk02,YS10A,BC01B}.  Nakanishi et al. \cite{nihk02} considered  quantum  finite automata of Kondacs-Watrous type \cite{kw97}, which use measurement at each step of the computation. 
They showed that  (unlike the case of classical finite automata) the class of languages recognizable by NQFAs properly contains the class of all regular languages. A full characterization of the class of languages recognized by all NQFA variants that are at least as general as the Kondacs-Watrous type was presented in \cite{YS10A}. It was shown that they define the class of exclusive stochastic languages.  

Bertoni and Carpentieri \cite{BC01B} considered a weaker model  -- nondeterministic quantum automata of Moore-Crutchfield type \cite{MC00} with a single measurement at the end of a computation.  They showed that the class of languages recognizable by this model does not contain any finite nonempty language but contains a nonregular language. 

In this paper we investigate nondeterministic quantum OBDDs where the model can evolves unitarily, followed by a projective measurement at the end. We call the model as nondeterministic unitary OBDD (NUOBDD). It can be seen as OBDD counterparts of unitary space bounded curcuits \cite{FL16A} or Moore-Crutchfield (measure-once) quantum finite automata \cite{MC00,AY15}. 

Section \ref{Sec:preliminaries} presents the necessary background. We present our results in Section \ref{Sec:main}. We start by presenting a generic lower bound on the widths of NUOBDD based on the size of strong 1-fooling sets (Section \ref{Sec:lower-bound}). Then, we present (i) new quantumly cheap but classical expensive functions by improving the previous gap (Section \ref{Sec:notPerm}) and (ii) classically cheap functions that are expensive for NUOBDDs (Section \ref{Sec:exactk}). We also present a function for which neither classical nor unitary nondeterminism does help (Section \ref{Sec:MOD}). Moreover, based on our results, we present a width hierarchy for NUOBDDs (Section \ref{Sec:hierarhy}). Lastly, we provide the bounds on the widths of NUOBDDs for the basic Boolean operations negation, union, and intersection (Section \ref{Sec:Union_Intersection_Complementation}). We close the paper by Section \ref{Sec:conclusion}.

\section{Preliminaries}
	\label{Sec:preliminaries}
In this section, we provide the necessary background to follow the remaining parts.
First, we give the definitions of the models.
Then, we present some basic facts from linear algebra which will be used in the
proofs.

\subsection{Definitions}

We use superscripts for enumerating  vectors and strings, and, subscripts for enumerating  the elements of vectors and strings.
A $d$-state quantum system ($QS$) can be described by a $d$-dimensional Hilbert space ($ \mathcal{H}^d $) over the field of complex numbers with the norm $||\cdot||_2$. A pure (quantum) state of the $QS$ is described by a column vector  $ \ket{\psi} \in \mathcal{H}^d $, whose length is one (unitary ket-vector), i.e.$ \sqrt{\braket{\psi}{\psi}} =1 $. 
As long as it is a closed system, the evolution of the $QS$ is described by some unitary matrices $U$. In order to retrieve information from the system, we can apply a projective measurement (then the system is no longer closed). We refer the reader to \cite{SayY14} for more details on the finite dimensonal QSs (see \cite{NC00} for a complete reference on quantum computing).

A branching program  (BP) on the variable set $X=\{x_1,\dots, x_n\}$ is a finite directed acyclic graph with one source node and sink nodes partitioned into two sets -- $Accept$ and $Reject$. Each
non-sink node is labelled by a variable $x_i$ and has two outgoing
edges labelled $0$ and $1$, respectively.  An input $\sigma$ is {\em
accepted} if and only if it induces a chain of transitions leading to
a node in $Accept$, otherwise $\sigma$ is rejected.

A BP $ P $ computes a Boolean function $f:\{0,1\}^n \rightarrow \{0,1\}$ iff  $ P $ accepts all $\sigma \in f^{-1}(1)$  and $P$ rejects all $\sigma\in f^{-1}(0)$.

A BP is {\em oblivious} if its nodes can be partitioned
into levels $V_0, \ldots, V_\ell$  such that
 nodes in $V_{\ell}$ are sink nodes, nodes in each level
$V_j$ with $0\le j < \ell$ have outgoing edges only to nodes in the next
level $V_{j+1}$, and all nodes in the level $V_j$ query the same
bit $\sigma_{i_{j+1}}$ of the input. If on each computational path from the source node to a sink node each variable from $X$ is tested at most once,  then such BP is  called \textit{read-once} BP. 

In this paper, we investigate read-once oblivious BPs that are commonly called as {\em Ordered Binary Decision Diagrams} (OBDDs). Since the lengths of OBDDs are fixed, the main complexity measure for them is their widths, i.e. for OBDD $P$, $width(P)=\max _{j} |V_j|$. The width of OBDDs can be seen as the number of states of finite automata and so we can refer the widths also as the sizes of OBDDs.

A nondeterministic OBDD (NOBDD) can have the ability of making more than one outgoing transition for each tested input bit from each node and so the program can follow more than one computational path and if one of the paths ends with an accepting node, then the input is accepted. Otherwise (all computation paths end with some rejecting nodes), the input is rejected.

Quantum OBDDs (QOBDDs) are non-trivial generalizations of classical OBDDs \cite{AGKY16A} when using general quantum operators like superoperators \cite{Wat09A}. Here we focus on a restricted version of QOBDDs that evolves only unitarily followed by a projective measurement at the end \cite{agk01}.: unitary OBDDs (UOBDDs).

Formally a \uobdd~$ M_n $, defined on the variable set $X=\{x_1,\dots, x_n\}$, with width $d$ (operating on $ \mathcal{H}^d $) is a quadruple
\[
	M_n=\bigl(Q, \ket{\psi^0},T, Q_{acc} \bigr ),
\]
where $ Q = \{q_1,\ldots,q_d\} $ is the set of states such that the set $ \{ \ket{q_1},\ldots,\ket{q_d} \} $ form a basis for $ \mathcal{H}^d $,  $\ket{\psi^0}\in \mathcal{H}^d$ is the initial quantum state,
$Q_{acc} \subseteq Q $ is the set of accepting states, and $T=\{ ( i_j, U_j(0), U_j(1)) \}_{j=1}^n$ is a sequence of  instructions such that  $i_j$ determines a variable $x_{i_j}$ tested at the step $j$, $U_j(0)$ and $U_j(1)$ are unitary
transformations defined over $\mathcal{H}^d$. 

For any given input $ \sigma \in  \{0,1\}^n $, the computation of $ M $ can be traced by a unitary vector, which is initially  
$ \ket{\psi^0} $. At the $j$-th step ($j=1, \dots, n$) the input bit $ x_{i_j} $ is tested and then the corresponding unitary operator is applied:
\[
	\ket{\psi^j} = U_{j}(\sigma_{i_j}) \ket{\psi^{j-1}},
\]
where $ \ket{\psi^{j-1}} $ and $ \ket{\psi^j} $ represent the quantum states after the $ (j-1)^{th} $ and $ j^{th} $ steps, respectively.

After all input bits are read, the following projective measurement is applied: 
$ P = \{ P_{acc},P_{rej}  \} $, where both $ P_{acc} $ and $ P_{rej} $ are diagonal 0-1 matrices such that $P_{acc}[j,j]=1$ iff  $q_j \in Q_{acc} $ and $ P_{rej} = I - P_{acc} $. Here $ P_{acc} $ ($P_{rej}$) projects any quantum state into the subspace spanned by accepting (non-accepting/rejecting) basis states. Then, the accepting probability of $ M_n $ on $ \sigma $ is calculated from the final state vector $\ket{\psi^n}$ as follows:
\[
	Pr^{M_n}_{accept}(\sigma)=||P_{acc}\ket{\psi^n}||^2.
\]

It is clear that $M_n$ defines a probability distribution over the inputs from $ \{0,1\}^n $. By picking some threshold between 0 and 1, we can classify the inputs as the ones accepted with probability greater than  the threshold and the others. Picking threshold as 0 is a special case and also known as nondeterministic acceptance mode for probabilistic and quantum models \cite{ADH97,YS10A}.

Nondeterministic \uobdd ~(\nuobdd) is a \uobdd, say $ N_n $, that is restricted to compute the Boolean function $f$ with threshold 0: each member of $ f^{-1}(1) $ is accepted with non-zero probability by $N_n$ and each member of $ f^{-1}(0) $ is accepted with zero probability by $N_n$. Then we say that $ f $ is computed by \nuobdd~ $ N_n $. 

A probabilistic OBDD (POBDD) $P_n$ can be defined in the same way as \uobdd~ $M_n$ with the following modifications: the initial state is a stochastic vector $ (v^0) $, each transformation is a stochastic matrix (the ones at the $j$-th levels are $ A_j(0) $ and $ A_j(1) $). Then, the computation is traced by a stochastic vector: at the $j$-th step ($j=1, \dots, n$) the input bit $ x_{i_j} $ is tested and then the corresponding stochastic operator is applied:
\[
	v^j = A_{j}(\sigma_{i_j}) v^{j-1},
\]
where $ v^{j-1} $ and $ v^j $ represent the probabilistic states after the $ (j-1)^{th} $ and $ j^{th} $ steps, respectively. Lastly, the accepting probability is calculated from the final vector as follows:
\[
	Pr^{P_n}_{accept}(\sigma)= \sum_{q_i \in Q_{acc}} v^n_i.
\]
If the initial probabilistic state and each stochastic matrix in $P_n$ is restricted to have only 0s and 1s, then all the computations become deterministic and so $ P_n $ is called a deterministic OBDD. If we do the same restriction to $ M_n $, then we obtain again a deterministic OBDD but its computation must be reversible (0-1 unitary matrices are also known as permutation matrices) and so it is called a (classical) reversible OBDD (ROBDD). Similar to quantum nondeterminism, $ P_n $ with threshold 0 forms an NOBDD. Besides  a POBDD or \uobdd~ is called exact if it accepts any input with probability either 1 or 0. Then, the corresponding function is called to be computed exactly.

\newcommand{\cobddnd}{\mathsf{OBDD_n^d}}
\newcommand{\cnobddnd}{\mathsf{NOBDD_n^d}}
\newcommand{\cnuobddnd}{\mathsf{NUOBDD_n^d}}

The classes $ \cobddnd $, $ \cnobddnd $, and $ \cnuobddnd $ are formed by the Boolean functions defined on $ \{0,1\}^n $ that can be respectively computed by OBDDs, NOBDDs, and \nuobdd s with width at most $ d $.


\subsection{Some facts from Linear Algebra} \label{Linear}
Let  $V$ be a vector space over the field $\mathbb{C}$ of complex numbers with the norm $||\cdot||_2$. We denote by ${\bf 0}$ zero element  of $V$. Here are the properties of norm:

\begin{enumerate}
	\item $||\psi||=0 \Rightarrow \psi={\bf 0}$; 
	\item $\forall \psi,\phi \in V, ||\psi+\phi||\leq ||\psi||+||\phi||$ (triangle inequality); and,
	\item $\forall \alpha \in \mathbb{C}, \forall \psi \in V, ||\alpha\, \psi||=|\alpha|\cdot ||\psi||$.
\end{enumerate}

A set of vectors  $\Psi=\{\psi_1, \psi_2, \dots , \psi_d \}\in V$ is linearly dependent iff there are  $\alpha_1, \alpha_2, \dots, \alpha_d \in \mathbb{C}$ such that
$\alpha_1 \psi_1 + \cdots + \alpha_d \psi_d = {\bf 0}$ and $\alpha_j \neq 0$  for some $j \in \{1, \dots, d\}$. 
If $\alpha_1 \psi_1 + \cdots + \alpha_d \psi_d = {\bf 0}$ only when $\alpha_1 = \alpha_2 = \dots = \alpha_d = 0$, then the set $ \Psi  $ is linearly independent.

It is known that a  set of vectors $\Psi=\{\psi_1, \psi_2, \dots , \psi_d \}$ is linearly independent iff either $\Psi =\emptyset$ or $\Psi$ consists
of a single element $\psi \neq {\bf 0}$, or $|\Psi|\geq 2$ and no vector $\psi_j\in \Psi$ can be expressed
as a linear combination of the other vectors of $\Psi.$
When  the set $\Psi$ is not empty and is linearly independent then $\psi_j\neq {\bf 0}$ for all $j$ and no vectors $\psi_i,\psi_j $ are collinear. If a set $\Psi$ is linearly independent, then every $\Psi'$ ($\Psi' \subseteq \Psi$) is linearly independent.

If $\Psi$ is a set of linearly independent vectors and $\psi \notin \Psi$ can not be expressed as a linear combination of the vectors from  
$\Psi$, then the set $\Psi \cup \{\psi\}$ obtained by adding
$\psi$ to the set $\Psi$ is linearly independent.

\begin{lemma}\label{cor1}
Let $ \{\psi_1, \psi_2, \dots , \psi_d \}\in V $ be a linearly independent set of vectors and $ U $ be a unitary transformation of the space $V$. Then,  the set of vectors $ \{U\psi_1,$ $U\psi_2, \ldots ,$ $U\psi_d\} $ is linearly independent.
\end{lemma}

\proof
Assume the set $ \{\psi'_1, \psi'_2, \dots , \psi'_d: \psi'_j=U\psi_j, j=1, \dots, d \}$ is  linearly dependent. Then there are $\alpha_1, \alpha_2, \dots, \alpha_d \in \mathbb{C}$
such that
$\alpha_1 \psi'_1 + \cdots + \alpha_d \psi'_d = {\bf 0}$ and $\alpha_j \neq 0$  for some $j \in \{1, \dots, d\}$.  Because $U$  is a unitary transformation it is hold $U^{\dagger}U=I$, where $I$ is the identity matrix and $U^{\dagger}$ is the conjugate transpose of $U$.  By the linearity of  transformation we have
$\alpha_1 U^{\dagger}\psi'_1 + \cdots + \alpha_d U^{\dagger}\psi'_d = {\bf 0}$.  But $\alpha_1 U^{\dagger}\psi'_1 + \cdots + \alpha_d U^{\dagger}\psi'_d = \alpha_1 U^{\dagger}U\psi_1 + \cdots + \alpha_d U^{\dagger}U\psi_d= \alpha_1\psi_1 + \cdots + \alpha_d\psi_d=  {\bf 0}.$ This is a contradiction. 
\Endproof

\begin{lemma}\label{lem1} Let $\psi_1,\dots, \psi_m,\psi \in V$, and vectors $\psi_1,\dots, \psi_m$ are linearly independent. Let $U$ be a linear map in $V$ such that 
$
	||U{\ket{\psi_i}}||=0 
$
for $ i=1, \dots, m $ and
$
	||U{\ket{\psi}}||> 0
$.
Then the set $\{\psi_1,\dots, \psi_m, \psi\} $ is  linearly independent.
\end{lemma}

\proof Suppose that the set $\{\psi_1,\dots, \psi_m,\psi \}$ is linearly dependent. Then there are   $\alpha_1, \alpha_2, \dots, \alpha_m \in \mathbb{C}$ such that $\psi=\alpha_1 \psi_1 + \cdots + \alpha_m  \psi_m$ and $\alpha_j \neq 0 $ for some $j=1, \dots, m$.
By linearity of $U$ we have $U\psi=U(\alpha_1 \psi_1 + \cdots + \alpha_m \psi_m)=\alpha_1 U\psi_1 + \cdots + \alpha_m  U\psi_m$.
Using the properties 2 and 3 of norm we have
$$||U\psi||\leq|\alpha_1|\cdot ||U\psi_1|| + \cdots + |\alpha_m| \cdot ||U\psi_m||.$$
Since by the hypothesis $||U\psi_1||= \cdots = ||U\psi_m||=0$ then we have $||U\psi||=0$. This is a contradiction. 
 \Endproof

\section{Main results}
\label{Sec:main}

In this section we present our results under six subsections.   It has already been known that nondeterministic quantum  OBDDs can be more efficient than classical ones.  In \cite{AGKY14} some functions were presented that are computed by \nuobdd s with constant width but NOBDDs need at least logarithmic width ($\Omega(\log n)$).  In Section \ref{Sec:notPerm}, we present an example of Boolean function based on which we can obtain a better superiority result.

\subsection{A lower bound for NUOBDDs}\label{Sec:lower-bound}

Let $f:\{0,1\}^n\rightarrow\{0,1\}$ be an arbitrary function and $\pi=(i_1,\ldots , i_n)$ be a permutation of $\{1, \ldots, n\}$. For a given $X=\{x_1,\dots, x_n\}$, an integer $k$ ($0<k<n$) and a permutation $\pi$, $X^{\pi}_k $ denotes $\{x_{i_1},\dots, x_{i_k}\}$. Any possible assignment on $X^{\pi}_k $, say $\sigma\in \{0,1\}^k$, is denoted by $\rho_{\pi, k}^{\sigma}:X^{\pi}_k \rightarrow \sigma$. Then $f|_{\rho_{\pi, k}^{\sigma}}$ is called  a subfunction obtained from $f$ by applying $\rho^{\sigma}_{\pi, k}$.

A set  $S^{\pi}_k=\{(\sigma,\gamma): \sigma\in \{0,1\}^k, \gamma\in \{0,1\}^{n-k}\}$ is called  a {\em strong 1-fooling set}  for $f$ if 
\begin{itemize}
\item $f|_{\rho^{\sigma}_{\pi, k}}(\gamma)=1$ for each $(\sigma, \gamma) \in S^{\pi}_k$ ,
\item if $(\sigma, \gamma),(\sigma', \gamma')\in S^{\pi}_k $, then  $f|_{\rho^{\sigma}_{\pi, k}}(\gamma')=0$ and $f|_{\rho^{\sigma'}_{\pi, k}}(\gamma)=0$.
\end{itemize} 
 
Let $\sigma, \sigma' \in \{0,1\}^k$.   We say that the string  $\gamma \in \{0,1\}^{n-k}$ distinguishes 
the string $\sigma$ from the string $\sigma'$, if  $f|_{\rho^{\sigma}_{\pi, k}}(\gamma) > 0 $ and $f|_{\rho^{\sigma'}_{\pi, k}}(\gamma)=0.$ Note that this definition is not symmetric. 

\begin{theorem} 
	\label{fooling_set} 
	Let \nuobdd~ $N_n$ computes a function $f:\{0,1\}^n\rightarrow \{0,1\}$  reading variables in an order $\pi=(i_1,\dots , i_n)$. Then $$Width(N_n) \geq \max_k|S^{\pi}_k|.$$
\end{theorem}
\begin{proof}
	Let $d=\max_k|S^{\pi}_k|$ and let $l$ be an index providing $|S^{\pi}_{l}|=d $, and $S^{\pi}_l=\{(\sigma^1,\gamma^1),\ldots,(\sigma^{d},\gamma^{d})\}$.  Consider the $l$-th level of $N_n$. 
Let  $\Psi=\{\ket{\psi(\sigma^j)}\, | \, j=1, \ldots, d\}$ be a set of state vectors of program
 $N_n$ after processing inputs  $\sigma^1, \ldots, \sigma^{d}$, i.e. $\ket{\psi(\sigma^j)}=U(\sigma^j)\ket{\psi^0}.$

\begin{claim}
	The set $\Psi$ is  linearly independent. 
\end{claim}
\begin{proof} 
	Assume that $\Psi$ is not linearly independent. Then there is a quantum state $\ket{\psi} =\ket{\psi(\sigma^i)} \in \Psi$  expressed as a linear combination of the others in $\Psi$:
\[
	\ket{\psi(\sigma^i)}=\sum_{j=1\atop j\neq i}^{d}\alpha_j \ket{\psi(\sigma^j)},
\]
and  $\alpha_j \neq 0 $ for some $j$.

Let $\gamma^i$ be a string such that $(\sigma^i, \gamma^i)\in S^{\pi}_l$. Then, by definition, for every input $\sigma^j$ $(j\neq i)$, we have  $f|_{\rho^{\sigma^j}_{\pi, k}}(\gamma^i)=0$, and program $N_n$ accepts the inputs $\sigma^j\gamma^i$ with zero probability:
 $$Pr_{accept}^{N_n}(\sigma^j\gamma^i)=||P_{acc}U(\gamma^i)\ket{\psi(\sigma^j)}||^2=0.$$
That means $||P_{acc}U(\gamma^i)\ket{\psi(\sigma^j)}||=0$.

The final quantum state for the input $\sigma^i\gamma^i$ is
\[
	\ket{\psi(\sigma^i\gamma^i)}=U(\gamma^i)\ket{\psi(\sigma^i)}=U(\gamma^i)\sum_{j=1\atop j\neq i}^{d}\alpha_j \ket{\psi(\sigma^j)}
\]
and by linearity we can follow that 
\[
	\ket{\psi(\sigma^i\gamma^i)}=\sum_{j=1\atop j\neq i}^{d}\alpha_j U(\gamma^i)\ket{\psi(\sigma^j)}=
\sum_{j=1\atop j\neq i}^{d}\alpha_j \ket{\psi(\sigma^j\gamma^i)}.
\]
Then, the accepting probability of the input $\sigma^i\gamma^i$ can be calculated as
\[
	Pr_{accept}^{N_n}(\sigma^i\gamma^i) =||P_{acc}\ket{\psi(\sigma^i\gamma^i)}||^2 =
\]
\[
	||\sum_{j=1\atop j\neq i}^{k}\alpha_j P_{acc} \ket{\psi(\sigma^j\gamma^i)}||^2 \leq 
(\sum_{j=1\atop j\neq i}^{k}|\alpha_j| \,|| P_{acc} \psi(\sigma^j\gamma^i)||)^2=0.
\]
However, $f|_{\rho^{\sigma^i}_{\pi, k}}(\gamma^i)>0$ and $N_n$ must accept this input with nonzero probability. Since this is a contradiction, the set $\Psi$ is linearly independent. $ \blacktriangleleft $
\end{proof}

Since the set  $\Psi$ of the state vectors of  $N_n$ at the  $l$-th level is linearly independent and its size is $d$ ($|\Psi|=d$), then  the dimension of the space of states of $N_n$  cannot be less than $d$: $Width(N_n)\geq d$.
\qed\end{proof}

\newcommand{\modpn}{\mathtt{MOD^p_{n}}}

\subsection{Function notPerm}
\label{Sec:notPerm}

Let $n=m^2$ for some $m>0$. We define function $\mathtt{notPERM_n}: \{0,1\}^{n} \rightarrow \{0,1\}$ as
\[ 
	\mathtt{notPERM_n}(\sigma) = \left\lbrace \begin{array}{lll}
		0 & , & \mbox{if } A(\sigma) \mbox{ is a permutation matrix,}\\
		1 & , & \mbox{otherwise,}
	\end{array}	 \right. 
\]
where the input bits are indexed as
$$
	x_{1,1},\dots,x_{1,m},x_{2,1},\dots,x_{2,m},\ldots,\dots,x_{m,1},\ldots,x_{m,m}
$$ and $x_{i,j}$ is $\sigma_{i,j}$, the $ (i,j) $-th entry of $ A $.
Note that $A$ is a permutation matrix if and only if it contains exactly one 1 in every row and in every column. 

The column and row summations of $ A $ can be represented by a $ 2m $ digit integer in base $ (m+1) $:
\[
	T(A) = (c_m c_{m-1} \cdots c_1 r_m r_{m-1} \cdots r_1), 
\]
where $ c_i $ and $ r_i $ are the summations of the entries in $i$-th column and $j$-th row, respectively, for $ 1 \leq i,j \leq m $. Then $ T(A) $ can be a value between 0 and $ T_{max} = (m+1)^{2m} -1 $, i.e. between $ (0 \cdots 0) $ and $ (m \cdots m) $. It can be easily verified that $ A $ is a permutation matrix if and only if $ T(A) = (1 \cdots 1) = \sum_{i=0}^{2m-1} (m+1)^i = T_{perm} $.

\begin{theorem}
	\label{quantum_notPerm}
	Function $\mathtt{notPERM_n}$ is computed by  a width-2 \nuobdd~$N_n$.
\end{theorem}
\proof 
	The  $ \nuobdd $ $N_n$ has two states $ \{q_1,q_2\} $, $ q_2 $ is the only accepting state, and $N_n$ operates on $ \mathbb{R}^2 $. Let $ \alpha $ be the angle of $ \frac{\pi}{T_{max}} $. The initial state is 
	\[
		\cos( - T_{perm} \alpha ) \ket{q_1} +
		\sin( - T_{perm} \alpha ) \ket{q_2},
	\]
	the point on the unit circle away from $ \ket{q_1} $ by angle $ T_{perm}(A)\alpha $ in clockwise direction. After reading the input, $ N_n $ makes a counter clockwise rotation with angle $ T(A) \alpha $, i.e., it rotates with angle $ \alpha\left( (m+1)^i + (m+1)^{m+j} \right) $ if $ x_{i,j} = 1 $ and it applies identity operator if $ x_{i,j} = 0 $. 
	
	If $ A $ is a permutation matrix, it makes a total rotation with angle $ T_{perm}\alpha $ and so the final quantum state becomes $ \ket{q_1} $. Thus, the input is accepted with zero probability. 
	
	If $ A $ is not a permutation matrix, then the amplitude of $ \ket{q_2} $ in the final quantum state always takes a nonzero value and so the input is always accepted with nonzero probability.	
	Note that $ N_n $ can make at most $ \pi $ degree rotation.
 \Endproof

It is known that function $\mathtt{PERM_n}$ ($ \neg \mathtt{notPERM_n}$) is not efficiently computed by classical read-once BPs, where $ \mathtt{PERM_n}(\sigma) = 1 $ iff $A(\sigma)$ is a permutation matrix. By using a known lower bound given for BPs, we can obtain a lower bound for NOBDDs solving $ \tt notPERM_n $.
 \begin{fact}
 	\cite{KMW91}
 	\label{KMW91} 
 	The size of any nodeterministic read-once BP, computing $\mathtt{PERM_n}$, cannot be less than $2^m/(2\sqrt{m})$, where $m=\sqrt{n}$.
 \end{fact}
  \begin{theorem}
 	\label{classical_notPerm} 
 	The width of any NOBDD computing $\mathtt{notPERM_n}$ cannot be less than $\sqrt{n}-\frac{5}{4}\log n -1$. 
 \end{theorem}
\proof  
	Since deterministic PB is a particular case of nondeterministic PB,  by Fact \ref{KMW91} we have that the size of any deterministic read-once PB computing $\mathtt{PERM_n}$ cannot be less than $2^m/(2\sqrt{m})$. Then, the size of any deterministic OBDD computing $\mathtt{notPERM_n}$ cannot be less than $2^m/(2\sqrt{m})$, too.  Having a lower bound for size, we can easily obtain a lower bound for width: since read-once PB has at least $n$ levels, then by the Pingeonhole principle we have that   $width(P)\geq size(P)/n$ for any read-once PB $P$. Next we can use the following  well-known relation between deterministic and nondeterministic space complexity: 
if a Boolean function $f$ is computed by an NOBDD of width $d$, then 	there exists a deterministic OBDD of size $2^d$ that computes $f$.  From this we conclude that any NOBDD, computing  $\mathtt{notPERM_n}$, has width at least $\log (2^m/(2n\sqrt{m}))$. Taking into consideration $n=m^2$   we get the lower bound for width of NOBDD computing  $\mathtt{notPERM_n}$.
\Endproof

Remark that any NOBDD can be simulated by a nondeterministic QOBDD with the same width if quantum model can use superoperators \cite{AGKY14}. However, as shown here, NOBDDs and \nuobdd s with the same widths are incomparable under certain bounds.

\subsection{Function EXACT}
	\label{Sec:exactk}

\newcommand{\exactk}{\mathtt{EXACT^k_n}}
\newcommand{\notexactk}{\mathtt{notEXACT^k_n}}
\newcommand{\andn}{\mathtt{AND_n}}

We continue with a classically cheap but \textit{unitarily} expensive function: 
$ \exactk: \{0,1\}^{n}\rightarrow \{0,1\}.$
	\[ \exactk(\sigma) = \left\lbrace 
	\begin{array}{lll}
		1 & , & if \,\#_1(\sigma)= k, \\
		0 & , & \mbox{otherwise, }\\
    \end{array}	 \right. \]
where $\#_1(\sigma)$ is a number of 1s in $\sigma$. If $k=n$, then  we have the function $\andn: \{0,1\}^{n}\rightarrow \{0,1\}$ that equals 1  iff the input does not contain any 0.
\begin{theorem}\label{quantum-upper-bound-exactk}
	There exists a \uobdd~$M_n$ with width $d=\max\{k+1, n-k+1\}$ that computes $ \exactk $ exactly (and so nondeterministically).
\end{theorem}
\begin{proof}
Assume that $k\geq n/2$. Then $d=k+1$. We design $ M_n $ as an ROBDD.
	Let $ \{ \ket{q_0},\ldots,\ket{q_{k-1}},\ket{q_{k}} \} $ be the basis states of $M_n$, $ \ket{q_0} $ is the initial quantum state, and $ \{ q_{k} \} $ is the only accepting state. When $M_n$ reads 0, the quantum state is not changed; and, when it reads 1, the quantum state  $ \ket{q_j} $ is changed to $ \ket{q_{j+1 \mod {(k+1)}}}$  for $ 0 \leq j \leq k $. So, if $M_n$ reads $ k $ 1s, the quantum state is set to $ \ket{q_{k}} $ and so the input is accepted with probability 1. Otherwise,  the input is accepted with probability 0. The property $k\geq n/2$ guaranties that $M_n$ can not visit $q_{k}$ twice.

If $ k < n/2$, $ M_n $ simply counts 0s instead of 1s in the above algorithm. 
\qed\end{proof}

\begin{theorem}
	\label{quantum-lower-bound-exactk}
	The width of any \nuobdd~ computing $\exactk$ cannot be less than  $ max\{k+1, n-k+1\} $.
\end{theorem}

\proof
Let $N_n=\bigl(Q, \ket{\psi^0},T, Q_{acc} \bigr )$ be an \nuobdd~ that computes $\exactk$, $\pi=(i_1,\dots , i_n)$ be an order of reading variables used  by  $N_n$, and $d=\max\{k, n-k\}$.

The computation begins from the initial configuration  $\ket{\psi^0}$. The input is of the form $\sigma=\sigma_1 \cdots \sigma_{n}$.  After the $l$-th step of the computation $(1 \leq l \leq n-1)$, the variables $x_{i_1}, \dots, x_{i_l}$ are read by $N_n$ and the configuration is $\ket{\psi^l(\sigma_{i_1}\cdots \sigma_{i_l})}$.  At the $(l+1)$-th  step, $N_n$ reads the next variable $x_{i_{l+1}}=\sigma_{i_{l+1}}$ and the new configuration becomes $\ket{\psi^{l+1}(\sigma_{i_1}\cdots \sigma_{i_l} \sigma_{i_{l+1}})}=U_{l+1}(\sigma_{i_{l+1}})\ket{\psi^l(\sigma_{i_1}\cdots \sigma_{i_l})}$. At the end of the computation, the projective measurement is applied to the resulting configuration $\ket{\psi^{n}(\sigma_{i_1}\cdots \sigma_{i_{n}})}$, and then, the probability of accepting the input is calculated as $Pr_{acc}^{N_n}(\sigma)=||P_{acc} \ket{\psi^{n}(\sigma_{i_1}\cdots \sigma_{i_{n}})}||^2.$

The idea behind our proof is as follows. For each level  $l$ ($l=0, \dots, d $) of  $ N_n $, we consider the set of all possible quantum states and then focus on a maximal subset that is linearly independent. Then we can give a lower bound on the size of this subset. 

Let $\Psi_l=\{ \ket{ \psi^l(\sigma) }: \sigma\in\{0,1\}^{l}\}$ be the set of all possible quantum states after the $l$-th step, i.e. $ \ket{ \psi^l(\sigma) } = U_l (\sigma_{l})\cdots U_1(\sigma_{1}) \ket{ \psi^0 } $. 


\begin{lemma}
	\label{lem2} 
	Let $ \ket{ \psi^1 } ,\dots, \ket{ \psi^m }, \ket{ \psi } \in \Psi_l$ and $ \ket{ \psi^1 }, \dots, \ket{ \psi^m } $ be linearly independent for some $ m \geq 1 $, where $ \ket{ \psi^i } = \ket{ \psi^l(\sigma^i) } $ for $i=1, \dots, m$ and $ \ket{ \psi } = \ket{ \psi^l(\sigma) } $.  If there exists a string $\gamma\in \{0,1\}^{n-l}$ that distinguishes the string $\sigma$ from each of the  strings  $\sigma^1,\dots, \sigma^m$, then the set $\{ \ket{ \psi^1},\dots, \ket{\psi^m}, \ket{\psi} \}$ is  linearly independent. 
\end{lemma}
\begin{proof}
	Let $U=U_n(\gamma_{n-l})\cdots U_{l+1}(\gamma_1).$ It is given that $|| P_{acc} U \ket{ \psi^i } ||=0$ for each $i=1, \dots, m,$ and $ || P_{acc} U \ket{ \psi } ||>0$.  Due to Lemma \ref{lem1}, we can follow that the set $\{ \ket{ \psi^1},\dots, \ket{\psi^m}, \ket{\psi} \}$ is linearly independent.
	$\blacktriangleleft$
\end{proof}


Let $\Phi_l$ ($\Phi_l\subseteq \Psi_l$) be the maximal set of linearly independent vectors.  We will estimate the cardinality of $\Phi_l$ by induction on $l$ ($l=0, \dots, d$).  We will consider two cases: when $k\geq n/2$ and when $k<n/2$.

{\bf Case 1}. First we assume $k\geq n/2$ that is $d=k$. 

Initial step: At the level $l=0$, the set $\Psi_0$ consists of a single vector $\ket{\psi^0}$. So we have  $|\Phi_0|=1$. At the level $l=1$, the set $\Psi_1$ contains two vectors $\ket{\psi^1(0)},\ket{\psi^1(1)} $. It is clear that these vectors are linearly independent since the string $\gamma = 1^{k-1}0^{n-k}$ distinguishes the string $1$ from the string $0$.

Induction step (for $l=2,\dots, d$): At the $(l-1)$-th level, we assume that $\Phi_{l-1} \subseteq \Psi_{l-1} $ has at least $ l $ elements, say $ \ket{ \psi^{j_0}},  \dots, \ket{ \psi^{j_{l-1}}} $, where the corresponding inputs are $\sigma^{j_0}, \dots, \sigma^{j_{l-1}} \in \{0,1\}^{l-1}$ respectively.
 
At the $l$-th step, $N_n$ reads the value $x_{i_{l}}=\sigma_{i_l}$. Due to Lemma \ref{cor1} (Section \ref{Linear}), we know that   the set $\Phi_{l}^0=\{U_{l}(0) \ket{ \psi^{j_0}  }, \ldots, U_{l}(0) \ket{ \psi^{j_{l-1}} } \}$  is linearly independent. It is clear that $ \ket{ \psi^l(1^l) } = U_{l}(1) U_{l-1}(1) \cdots U_{1}(1) \ket{\psi^0}  $ is not a member of $\Phi_{l}^0$. Moreover, the string $ 1^{k-l}0^{n-k} $ distinguishes $ 1^l $ from each of $ \sigma^{j_1}0, \dots, \sigma^{j_{l}}0 $. Therefore, due to Lemma \ref{lem2}, we can follow that the set $\Phi_{l}^0 \cup \{ \ket{ \psi^l(1^l) } \}$ is linearly independent. Thus, $\Phi_l $ contains at least $ (l+1) $ elements, i.e. $ \ket{ \psi^{j_0}},  \dots, \ket{ \psi^{j_{l}}},$ and $ \ket{ \psi^l(1^l) } $.

Therefore, $ \Phi_{d} $ has at least $ d+1 $ elements and so the dimension of quantum states must be at least $ d+1 $. 

{\bf Case 2}. Now assume that $k<n/2$ and therefore $d=n-k$. It is clear that $\exactk(\sigma)=1$ iff $\#_0(\sigma)=n-k$, where $\#_0(\sigma)$ denotes the number of 0s in $\sigma$ and  we have $n-k\geq n/2$. We can apply the same reasoning as in the previous case by interchanging 0 and 1. 

Therefore, in both cases $ \Phi_{d} $ has at least $ d+1 $ elements and so the dimension of quantum states must be at least $ d+1 $, where $d=\max\{k, n-k\}$.  Since there is a \nuobdd~ with width  $ (d+1) $ to solve $ \exactk $, we can also conclude that $ |\Phi_{d}| = d+1 $.
\Endproof

\begin{theorem} 
	\label{exact-classical-upper}
	The  function $\exactk$ is computed by an OBDD $D_n$ with width $\min(k+1,n-k+1)$ $+1$. 
\end{theorem}
\begin{proof}
Let $k< n/2$. The OBDD $D_n$ uses the order $\pi=(1,\ldots, n)$ and has  states $ q_0,\ldots, q_{k+1}, q_{k+2} .$ The initial state is $q_0$ and the only accepting state is $q_{k+1} $. $D_n$ counts number of 1s in the input moving from the current state $q_i$ to the state $q_{i+1}$ ($i=1, \ldots, k+1$) when reading 1, and does not changing the current state when reading 0. After entering the state $q_{k+2},$ that happens only when it reads the $(k+1)$-th 1, $D_n $ never leaves this state.  So, only for members of $\exactk$, $ D_n $ starts in $q_0$, reaches $q_{k} $ and stays there until the end, and so accepts the input. 

Let $k\geq n/2$, and then $n-k <n/2$. The OBDD $D_n$ is constructed in the same way by counting 0s instead of 1s. 
\qed\end{proof}

\begin{theorem} 
	\label{exactk-classical-lower}
	The width of any NOBDD computing $\exactk$ cannot be less than $\min(k+1,$ $n-k+1)+1$.  
\end{theorem}

\begin{proof}
Let $d=\min(k, n-k)$. Assume $k \leq n/2$ that is $d=k$. Let $P_n$ be an NOBDD that computes $\exactk$ and has width $<d+2$. Consider the $k$-th level $V_k$ of $P_n$ and a set of partial inputs $\Sigma=\{\sigma^j \in \{0,1\}^k : \sigma^j=$ $\underbrace{0\cdots 0}_{k-j}$ $\underbrace{1\cdots 1}_{j}, j=0, \ldots, k \}$. Let $ path(\sigma^j) $ be one of the paths after reading $ \sigma^j $ that can also lead the computation to an accepting node after reading $(k-j)$ more $1$s. Due to the Pigeonhole principle, each $ path(\sigma^j) $ must be in a different node of the $k$-th level and so $ V_k $ contains at least $ k+1 $ different nodes, say $ v_0,\ldots,v_k $.

The level $ V_{k+1} $ contains $ k+1 $ different nodes, say $ v'_0,\ldots,v_k' $, that can be accessed from $v_0,\ldots,v_k$ by reading a single 0, because from these nodes the computation can still go to some accepting nodes.  If a single 1 is read, then $ v_k $ must switched to a node other than $ v'_0,\ldots,v_k' $. If it switches to $ v_j' $, then the non-member input $ 1^k11^{k-j}0^* $ with length $ n $ is accepted since the computation from $ v_j' $ can go to an accepting node (the input $ 0^{k-j}1^j01^{k-j}0^* $ with length $ n $ is a member). Therefore, there must be at least $ (k+2) $ nodes.

If $k > n/2$, then  $n-k \leq n/2$ and so we can use the same proof by interchanging 0s and 1s. 
\qed\end{proof}

\begin{corollary}\label{and-classical-quantum}
The function $\andn$ is computed by an NOBDD $P_n$ with width 2.
The function $\andn$ is computed by an \nuobdd~ $N_n$ with width $n+1$ and there is no \nuobdd~ computing $\andn$ with width less than $n+1$.  

\end{corollary}

Now we show that negation of the function $\exactk$ is \textit{cheap} for \nuobdd. The Boolean function $ \notexactk: \{0,1\}^{n}\rightarrow \{0,1\}$ is defined as
	\[ \notexactk(\sigma) = \left\lbrace 
	\begin{array}{lll}
		0 & , & if \,\#_1(\sigma)= k, \\
		1 & , & \mbox{otherwise. }\\
    \end{array}	 \right. \]

\begin{theorem}\label{quantum-cheap}
	For any positive integer $k$ ($k\leq n$) the function $ \notexactk $ can be computed by an \nuobdd~ $ N_n $ with width 2.
\end{theorem}
\begin{proof}
	We use the same idea given in the proof of Theorem \ref{quantum_notPerm}. Let $ \alpha = \frac{\pi}{n} $. The \nuobdd~ $ N_n $ has two states $ \{q_1,q_2\} $, $ q_2 $ is the accepting state, and the initial quantum state is
	\[
		\cos(-k\alpha)\ket{q_1} + \sin(-k\alpha_1) \ket{q_2}.
	\] 
	After reading the input $\sigma$, $ N_n $ makes the counter clockwise rotation with angle $ k' \alpha  $, where $ k' = \#_1(\sigma) $, i.e. it rotates with angle $ \alpha $ for each 1. If $ k = k' $, then the final state is $ \ket{q_1} $ and so the input is accepted with zero probability. Otherwise, the accepting probability is always nonzero. Note that $ N_n $ can make a rotation with angle at most $ \pi $.
\qed\end{proof}


\subsection{Function MOD}
\label{Sec:MOD}

Here we present a series of results based on Boolean function $ \modpn : \{0,1\}^n \rightarrow \{0,1\}$, which is defined as:
\[ 
	\mathtt{MOD^p_{n}}(\sigma) = \left\lbrace \begin{array}{lll}
		1 & , & \mbox{if } \#_1(\sigma) \equiv 0 \pmod{p}, \\
		0 & , & \mbox{otherwise, }
	\end{array}	 \right.
\]
where $\#_1(\sigma)$ is the number of 1s in the input $\sigma$.

It is clear that $ \modpn $ can be solved by reversible OBDDs and so by exact \uobdd s with width $ p $.

\begin{theorem}
	\label{Upper_MOD_p} 
	There is a width-$p$ ROBDD $R_n$ computing the function $\mathtt{MOD^p_{n}}$.
\end{theorem}
 \proof  $R_n$ has $p$ states $s_0,\dots, s_{p-1}$ and $s_0$ is the initial state. $R_n$  deterministically counts number of 1s in the input by modulo $p$. If the input's bit is 1, $R_n$ goes from the state $s_i$ to the state $s_{i+1 \pmod p}$ and applies the identity transformation, otherwise.  $R_n$  accepts the input  iff the final state is $s_0$. It is clear that  transitions of $R_n$ are reversible and the width of $R_n$ is $p$.
\Endproof

Now, we show that nondeterminism does not help neither classically nor quantumly in order to solve $ \tt MOD_n^p $.
\begin{theorem}
	\label{MODp-upper} 
	If $p\leq n/2$, then the width of any NOBDD computing $\mathtt{MOD^p_{n}}$ cannot be less than $p$. 
\end{theorem}
\label{NOBDD_MOD}
\proof Assume that there exists an NOBDD $P_n$ that computes $\mathtt{MOD^p_{n}}$ and has  width $q < p$. Let $\Sigma=\{\sigma^1, \dots, \sigma^{p}:\sigma^j\in \{0,1\}^{n-p+1},\sigma^j=$ $\underbrace{0 \cdots 0}_{n-p+1-j}$ $\underbrace{1\cdots 1}_{j}\}$. Let $ path(\sigma^j) $ be one of the path after reading $ \sigma^j $ which also leads the computation to an accepting node later.  Since $Width(P_n)<p$,  by the Pigeonhole principle, there exist $ path(\sigma^i) $ and $path(\sigma^j)$, corresponding to $\sigma^i $ and $ \sigma^j $ respectively, that have the same node at the $(n-p+1)$-th level. It is clear that from this node the computation ends in an accepting node after reading $(p-i)$ $1$s. More specifically, the inputs $ \sigma^i \underbrace{0 \cdots 0}_{i-1} \underbrace{1 \cdots 1}_{p-i} $ and $ \sigma^j \underbrace{0 \cdots 0}_{i-1}\underbrace{1 \cdots 1}_{p-i} $ are accepted by $ P_n $. Since the second string must be rejected by $ P_n $, it is a contradiction. 
\Endproof

\begin{theorem}
	\label{Lower-Mod_p}
	For any $p$ ($p\leq n$)  the width of any \nuobdd~ computing $\mathtt{MOD^p_{n}}$ cannot be less than $p$.
\end{theorem}
\proof 
Let $p\leq n/2$. 
For any order $\pi$ of reading variables we can construct the following strong 1-fooling set for the function $\mathtt{MOD^p_{n}}$: 
$$S^{\pi}_{n-p+1}=\{(\sigma^i,\gamma^i): i=0, \ldots, p-1,\sigma^i=\underbrace{0\cdots 0}_{n-p+1-i} \underbrace{1\cdots 1}_{i},\gamma^i=\underbrace{0\cdots 0}_{i-1} \underbrace{1\cdots 1}_{p-i} \}.$$ Due to Theorem \ref{fooling_set}, we follow the result.

Let consider the case  $p>n/2$. Using  the same arguments as in the proof of Theorem \ref{quantum-lower-bound-exactk} we can show that on  the  $(p-1)$-th level the set of linear independent vectors, which are  achievable quantum states, contains at least $p$ elements. They are $\ket{\psi(\sigma^0)},\dots, \ket{\psi(\sigma^{p-1})}$, where $\sigma^j=1^{j}0^{p-j-1}, j=0, \dots, p-1. $ 
\Endproof

Currently we do not know whether using more general QOBDD models can narrow the width for $ \modpn $.


\subsection{Hierarchy for NUOBDDs}\label{Sec:hierarhy}

In  \cite{AGKY14,AGKY16A}, the following width hierarchy for OBDDs and NOBDDs  was presented. For any integer $n>3$ and $1< d \leq \frac{n}{2}$, we have
\[
	\mathsf{OBDD_n^{d-1}}  \subsetneq  \cobddnd
	~~ \mbox{ and }~~
	\mathsf{NOBDD_n^{d-1}} \subsetneq \cnobddnd.
\]
For any integer $n$, $d=d(n)$, $16\leq d \leq 2^{n/4}$, we have
\[
\mathsf{OBDD^{\lfloor d/8 \rfloor-1}}  \subsetneq  \mathsf{OBDD^{d}} \mbox{ and }\\
\mathsf{NOBDD^{\lfloor d/8 \rfloor-1}}  \subsetneq  \mathsf{NOBDD^{d}}.
\]
Here we obtain a complete hierarchy result for \nuobdd s with width up to $n$.

\begin{theorem}
For any integer $n>1$ and $1< d \leq n$, we have
	\[
		\mathsf{NUOBDD_n^{d-1}} \subsetneq \cnuobddnd.
	\]
\end{theorem}
 
\proof It is obvious that $\mathsf{NUOBDD^{d-1}}\subseteq \mathsf{NUOBDD^{d}}$. 
If $ d \leq n/2 $, we know that $\mathtt{MOD^d_{n}}\in\mathsf{NUOBDD_n^{d}} $ and $\mathtt{MOD^d_{n}}\notin\mathsf{NUOBDD_n^{d-1}} $ due to Theorems \ref{Upper_MOD_p} and \ref{Lower-Mod_p}.
If $ d>n/2 $, we know that $\mathtt{EXACT_n^{d-1}}\in\mathsf{NUOBDD_n^{d}} $ and  $\mathtt{EXACT_n^{d-1}}\notin\mathsf{NUOBDD_n^{d-1}} $ due to Theorems \ref{quantum-upper-bound-exactk} and \ref{quantum-lower-bound-exactk}.
\Endproof

\begin{theorem}
	(1) For any pair $ (d_1,d_2) $ satisfying  $  1 < d_1,d_2 \leq n $, $ \mathsf{NOBDD_n^{d_2}} \not\subseteq \mathsf{NUOBDD_n^{d_1}}$. (2) For any   $ (d_1,d_2) $ satisfying $  1 < d_1, d_2 < \sqrt{n}-\frac{5}{4}\log n -1$, $ \mathsf{NUOBDD_n^{d_2}} \not \subseteq \mathsf{NOBDD_n^{d_1}} $.
\end{theorem}
\begin{proof}
Let $d_1, d_2$ be arbitrary integers satisfying  $  1 < d_1,d_2 \leq n $. 
By Corollary \ref{and-classical-quantum}, we know that $\andn \in \mathsf{NOBDD_n^2} \subseteq \mathsf{NOBDD_n^{d_2}} $ and   $\andn \not\in \mathsf{NUOBDD_n^{n}} $ and so $ \andn \not\in \mathsf{NUOBDD_n^{d_1}} $. Therefore, $ \mathsf{NOBDD_n^{d_2}} \not\subseteq \mathsf{NUOBDD_n^{d_1}}$.

Let $d_1, d_2$ be arbitrary integers satisfying  $  1 < d_1,d_2 <\sqrt{n}-\frac{5}{4}\log n -1$.  
By Theorem \ref{quantum_notPerm} and Corollary \ref{classical_notPerm}, we know that $\mathtt{notPERM_n} \in \mathsf{NUOBDD_n^2} \subseteq \mathsf{NUOBDD_n^{d_1}} $ and   $\mathtt{notPERM_n} \not\in $ $\mathsf{NOBDD_n^{d_2}} .$ Therefore, $ \mathsf{NUOBDD_n^{d_2}} \not \subseteq \mathsf{NOBDD_n^{d_1}} $.
\qed\end{proof}

\subsection{Union, Intersection, and Complementation}\label{Sec:Union_Intersection_Complementation}

Let $f, g: \{0,1\}^n\rightarrow\{0,1\}$. We call a function $h=f\cup g$ the union of the functions $f$ and $g$ iff $h(\sigma)= f(\sigma)\bigvee g(\sigma)$ for all $\sigma\in \{0,1\}^n.$ We call a function $h=f\cap g$ the intersection of the functions $f$ and $g$ iff $h(\sigma)=f(\sigma)\bigwedge g(\sigma)$ for all $\sigma\in \{0,1\}^n.$ We call $h$ the negation of the function $f$ iff $h(\sigma)=\neg f(\sigma)$ for all $\sigma\in \{0,1\}^n.$

\begin{theorem}\label{Union}
	Let $ f $ and $ g $ are Boolean functions defined on $ \{0,1\}^n $ computed by an \nuobdd~ $N_n$ with width $ c $ and an \nuobdd~ $ N_n' $ with width $ d $ respectively such that $N_n$ and $N_n'$ use the same order $\pi$ of reading variables. Then, the Boolean function $ f \cup g $ can be computed by an \nuobdd, say $ N_n'' $, with width $ c + d $. 
\end{theorem}

\begin{proof}
Let $ N_n = ( Q = \{ q_1,\ldots,q_c \}$, $\ket{\psi^0},$ $T,$ $Q_{acc} ) $, 	$ N_n' =$ $( Q' =$ $\{ q'_1,\ldots,$ $q'_d \},$ $\ket{\psi'^0},$ $T',$ $Q'_{acc} ) $,  where 
$T=\{(i_j,$ $U_j(0),$ $U_j(1))\}_{j=1}^n$,   $T'=$ $\{(i_j,$ $U'_j(0),$ $U'_j(1))\}_{j=1}^n$. 
The  \nuobdd~ $ N''_n $ can be constructed based on $N_n$ and $N_n'$ as follows.

$ N''_n = ( Q''=Q\cup Q' = \{q_1,\ldots,q_c, q'_1,\ldots,$ $q'_d\},\ket{\psi''^0}, T'', Q''_{acc}=Q_{acc}\cup Q'_{acc} ),$ where the initial quantum state is $\ket{\psi''^0}= \frac{1}{\sqrt{2}}( \ket{\psi^0}\oplus\ket{\psi'^0})$. The sequence of instructions $T''=\{i_j, U''_j(0), U''_j(1)\}_{j=1}^n$, where
$U''_j(\sigma)=
\begin{pmatrix}
U_j(\sigma)& \mathbf{0}\\
\mathbf{0} & U'_j(\sigma)
\end{pmatrix}
$. Here $\mathbf{0}$ denotes zero matrix. 

By  construction, $ N_n'' $ executes both $ N_n $ and $ N_n'' $ in parallel with equal amplitude, and so it accepts a given input with zero probability iff both $N_n$ and $N_n'$ accept it with zero probability. In other words, it accepts an input with non-zero probability iff $ N_n $ or $N_n'$ accepts it with zero probability. Thus, $N_n''$ computes the function $f \cup h$.
\qed\end{proof}

\begin{theorem}
	\label{Intersection}
	Let $ f $ and $ g $ are Boolean functions defined on $ \{0,1\}^n $ computed by an \nuobdd~$N_n$ with width $ c $ and an \nuobdd~ $ N_n' $ with width $ d $, respectively, such that $N_n$ and $N_n'$ use the same order $\pi$ of reading variables. Then, the Boolean function $ f \cap g $ can be computed by an \nuobdd, say $ N_n'' $, with width $ c \cdot d $. 
\end{theorem}
\begin{proof}
Let $ N_n = ( Q = \{ q_1,\ldots,q_c \}$, $\ket{\psi^0},$ $T,$ $Q_{acc} ) $, 	$ N_n' =$ $( Q' =$ $\{ q'_1,\ldots,$ $q'_d \},$ $\ket{\psi'^0},$ $T',$ $Q'_{acc} ) $,  where 
$T=\{(i_j,$ $U_j(0),$ $U_j(1))\}_{j=1}^n$,   $T'=$ $\{(i_j,$ $U'_j(0),$ $U'_j(1))\}_{j=1}^n$. 
The  \nuobdd~ $ N_n'' $ can be constructed by tensoring $N_n$ and $N_n'$ as follows.

$ N_n'' = ( Q''=Q \times Q' = \{ q_{1,1},\ldots, q_{c,d} \},\ket{\psi^0}\otimes\ket{\psi'^0}, T'', Q'_{acc} ),$ where the sequence of instructions $T''=\{i_j, U_j(0)\otimes U'_j(0), U_j(1)\otimes U'_j(1)\}_{j=1}^n$ and the set of accepting states contains all the states $q_{i,j}$ satisfying $q_i\in Q_{acc}$ and $q_j\in Q'_{acc}.$

From this construction it follows that  $Pr^{N_n''}_{accept}(\sigma)=Pr^{N_n}_{accept}(\sigma)\cdot Pr^{N_n'}_{accept}(\sigma)$. If the input $\sigma$ is satisfying that $f(\sigma)=1$ and $g(\sigma)=1$, then both $N_n$ and $N_n'$ accept it with nonzero probability and therefore $N_n''$ also accepts this input with nonzero probability. If the input $\sigma$ is satisfying that $f(\sigma)=0$ or $g(\sigma)=0$ then $Pr^{N_n''}_{accept}(\sigma)=0$ for this input. 
\qed\end{proof}

The bound for intersection can be shown to be tight in certain cases. 

\begin{theorem}\label{MOD_pr}
	There exist functions $f$ and $g$ computed by {$\nuobdd$}s $N_{f,n}$ with width $c$ and $N_{g,n}$ with width $d$, respectively, such that the width of any  {$\nuobdd$}  computing the function $h=f \cap g$ cannot be less than $ lcm(c \cdot d)$, where $ lcm(c \cdot d) \leq n $.
\end{theorem}
\begin{proof}
	By Theorems \ref{Upper_MOD_p} and \ref{Lower-Mod_p}, we can follow the result. The functions $ \mathtt{MOD_n^c} $ and $ \tt MOD_n^d $ are computed by $\nuobdd$s with widths $ c $ and $ d $, respectively. Their intersection function is $ \tt MOD_n^l $, where $ l = lcm(c,d) $, and so the width of any $\nuobdd$ cannot be less than $ l $.
\qed\end{proof}

The bounds given in Theorems \ref{Union} and \ref{Intersection} are also valid for NOBDDs. Deterministic OBDDs, on the other hand, requires $c\cdot d$ for union operation.

Classically, if a function, say $f$, solved by an NOBDD with width $ d $, then the negation of $ f $ can be solved by another NOBDD with width at most $ 2^d $. By using Corollary \ref{and-classical-quantum} and the result below we conclude that in case of \nuobdd, we cannot provide such a bound.

\begin{corollary} (from Theorem \ref{quantum-cheap})
The function $\neg \andn$ is computable by NUOBDD with width 2.
\end{corollary}

\section{Concluding remarks}
\label{Sec:conclusion}
In this paper we investigate the width complexity of nondeteministic unitary OBDDs and compare them with its classical counterpart. Our results are mainly for linear and sublinear widths. As a future work, we plan to investigate the superlinear widths. Here we present a width hierarchy and a similar result is not known for nondeterministic quantum OBDDs using  general quantum operators. We also find interesting possible applications of our results to some other models like quantum finite automata.

\bibliographystyle{splncs03}
\bibliography{tcs}
\end{document}